\documentclass{amsart}

\usepackage{stackrel}
\usepackage[utf8]{inputenc}
\usepackage{amssymb}
\usepackage{amsthm}
\usepackage{tikz}
\usepackage{mathrsfs}
\usepackage{epsf}
\usepackage{bookmark,hyperref}
\hypersetup{
     colorlinks=true,
     linkcolor=blue,
     filecolor=blue,
     citecolor = blue,
     urlcolor=cyan,
     }
\usepackage{pstricks}

\usepackage{amsmath,amssymb,amscd}
\usepackage{amsthm}
\usepackage{subfigure}
\usepackage{verbatim}
\usepackage{fancybox}
\usepackage{subfigure}
\usepackage[all]{xy}

\usepackage{mathtools}

\DeclarePairedDelimiter\bra{\langle}{\rvert}
\DeclarePairedDelimiter\ket{\lvert}{\rangle}
\DeclarePairedDelimiterX\braket[2]{\langle}{\rangle}{#1 \delimsize\vert #2}

\newcommand{\R}{{\mathbb R}}
\newcommand{\C}{{\mathbb C}}

\newcommand{\I}{{\mathcal{O}}}

\newtheorem{teo}{Theorem}[section]
\newtheorem{lema}[teo]{Lemma}
\newtheorem{cor}[teo]{Corollary}

\newtheorem{prop}[teo]{Proposition}

\newtheorem{defi}{Definition}[section]

\makeatletter
\@namedef{subjclassname@2020}{%
  \textup{2020} Mathematics Subject Classification}
\makeatother

\emergencystretch=\maxdimen
\begin{document}

\title[Quantum network reliability]{Quantum network reliability with perfect nodes}

\author{J. M. Burgos}
\address{Departamento de Matem\'aticas, CINVESTAV-\,CONACYT, Av. Instituto Polit\'ecnico Nacional 2508, Col. San Pedro Zacatenco, 07360 Ciudad de M\'exico, M\'exico.}
\email{burgos@math.cinvestav.mx}
\thanks{The first author is a research fellow at \textit{Consejo Nacional de Ciencia y Tecnolog\'ia (CONACYT)}, Mexico.}

\begin{abstract}
We introduce the concept of quantum reliability as an extension of the concept of network reliability in the context of quantum networks. We show that this concept is intimately related to the concept of quantum reliability operator that we also introduce and show a Negami like splitting formula for it. Considering that the simple factorization formula for classical networks which is the basis of most of the calculation algorithms does not hold in the quantum context due to entanglement, a Negami like splitting for the quantum reliability operator becomes relevant.
\end{abstract}




\maketitle
\tableofcontents

\section{Introduction}

This paper is about the extension of the concept of network reliability in the context of quantum networks. The field of quantum network technology and quantum computing is rapidly evolving and realisations of these are expected in the near future. It is therefore necessary to have a notion of quantum network reliability that reproduces the classical notion on classical networks as well as having some splitting formula into smaller quantum networks to make the concept operational as in the classical case. For the state of the art regarding the quantum network realisation we recommend the survey \cite{Towards} written by W. Kozlowski and S. Wehner and references therein. For general concepts on quantum mechanics we recommend the classical book \cite{vonNeumann} by J. von Neumann. For the theory of quantum information and quantum computation we recommend the classical book \cite{Nielsen} by M. A. Nielsen and I. L. Chuang.

We start by recalling the standard concepts and notation of classical network reliability in section \ref{sectionI}. In section \ref{sectionII} we introduce the concept of \textit{quantum reliability function} on quantum networks which, in analogy with the classical reliability on classical networks, the main motivation for this function is the quantum network design. In this section we seek for a function that reproduces the classical network reliability on classical networks and behaves as expected on quantum networks under entangled states. After establishing the axioms for a quantum reliability function, Theorem \ref{main1} shows that it can be realized with a self adjoint operator that we call the \textit{quantum realiability operator}. It is important to remark this function is a priori independent of any interpretation of quantum mechanics and measurement hypotheses. However, under the assumption of \textit{non contextuality}, the value of the quantum reliability realized with the mentioned operator is the probability that the quantum network stays connected after measurement. We refer to the paper \cite{noncontext} by S. Kochen and E. P. Specker for the notion of non contextuality. 

Quantum states of a quantum network have the phenomenon of \textit{entanglement} and this invalidates the application of well known simple factorization formula for the exact calculation of the classical network reliability in the quantum context. Considering that this formula is the basis for most of the known exact calculation algorithms, it is necessary to have some splitting formula for quantum reliability in order to make it operational and functional. In section \ref{sectionIII}, Theorem \ref{main2} shows that a Negami like splitting formula holds for the quantum reliability operator.

In section \ref{sectionIV} we treat the case of \textit{hybrid classical-quantum networks}. These are important for technological applications as well as theoretical ones. One of the theoretical applications is the fact that a classical network with imperfect nodes can be effectively replaced by a hybrid network with perfect nodes.

An important particular case of a hybrid network is that of a \textit{classical network with a quantum sublayer}. These networks are important for technological applications as discussed in \cite{Towards}. We show in section \ref{sectionVI} that the quantum reliability of a classical network with a quantum sublayer is the sum of the classical reliability of the classical network with quantum corrections coming from the quantum sublayer.

The remaining sections contain the respective proofs of the mentioned results. Except for an algebraic combinatoric result concerning the invertibility of a matrix, the paper is self contained.

\section{Classical network reliability}\label{sectionI}

In reliability theory, we identify a (classical) network with a \textit{stochastic graph}. We will be interested in the case of perfect nodes. A stochastic graph is a pair $\left(G,\,(p_e)_{e\in E}\right)$ where $G$ is a graph and every edge $e\in E$ of it has associated a Bernoulli variable with parameter $p_e$ such that the variables are independently distributed. Identifying on every edge the \textit{one state} as an operating edge and the \textit{zero state} as a failure, we define the \textit{reliability} of $\left(G,\,(p_e)_{e\in E}\right)$ as the probability of the graph being connected. Concretely, denoting by
\begin{equation}\label{states_classical}
\Lambda_G\,=\,\{{\bf 0,1}\}^E
\end{equation}
the set of states of the stochastic graph and recalling that
\begin{equation}\label{rel1}
\mbox{Prob}(\epsilon)\,=\,\prod_{e\in E}\,p_e^{\epsilon(e)}\,(1-p_e)^{1-\epsilon(e)},\qquad \epsilon\in \Lambda_G
\end{equation}
due to the independent distribution of the variables, we have
\begin{equation}\label{rel2}
R\left(G,\,(p_e)_{e\in E}\right)\,=\,\mbox{Prob}\,\left([\epsilon\ \mbox{is}\ \mbox{connected}]\right)\,=\,\sum_{\epsilon\in \mathcal{C}}\, \mbox{Prob}(\epsilon)
\end{equation}
because of the independent character of the variables again. Here we have denoted by $\mathcal{C}$ the subset of connected states. Given an ordering of the edges, we denote the reliability of a stochastic graph by
$$R_G\,(p_1,\ldots,p_n)\,=\,R\left(G,\,(p_e)_{e\in E}\right).$$
Geometrically, the sequence of Bernoulli parameters $(p_e)_{e\in E}\,\in\,[0,1]^E$ lies in the cube whose vertices are the states in $\Lambda_G$, see expression \eqref{states_classical}. Equivalently, the sequence of Bernoulli parameters lies in the convex hull of set of states $\Lambda_G$. We have the reliability function
$$R_G:\,\mbox{Conv}\,(\Lambda_G)\rightarrow [0,1].$$

From equations \eqref{rel1} and \eqref{rel2}, it is clear that
\begin{equation}\label{SimpleFact}
R_G\,=\,p_e\,R_{G\cdot e}\,+\,(1-p_e)\,R_{G-e}
\end{equation}
where $G\cdot e$ denotes the graph resulting from removing the edge $e$ from $G$ and identifying the adjacent nodes and $G- e$ denotes the graph resulting form removing the edge $e$ form $G$, see Figure \ref{connect_QN}. We have abused of notation by omitting to write the dependence on the Bernoulli parameters. The equation \eqref{SimpleFact} is usually called the \textit{simple factorization formula} and it is the basis of most of the exact network reliability calculation algorithms. We recommend the book \cite{Colbourn} written by C. J. Colbourn and references therein for the theory of reliability of stochastic graphs.

\begin{figure}
\begin{center}
  \includegraphics[width=0.85\textwidth]{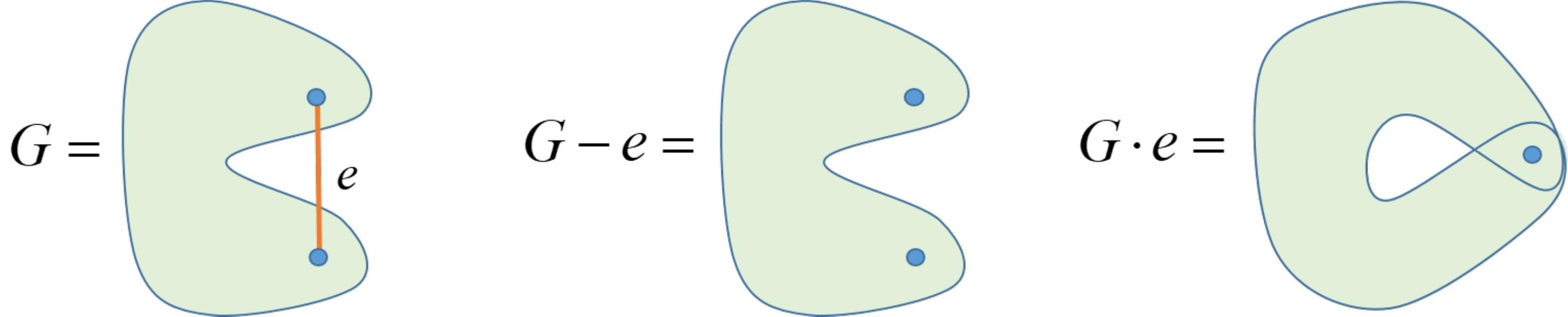}\\
  \end{center}
  \caption{Deletion and contraction of an edge.\ \ \ \ \ \ \ \ \ \ \ \ \ }\label{connect_QN}
\end{figure}

\section{Quantum network reliability}\label{sectionII}

Network reliability theory is an example of a classical system with a finite set of states \eqref{states_classical}. Although in general the quantization of a classical system is much more of an art than a well established algorithm, in the case of a finite set of classical states the situation is straightforward and it goes as follows \footnote{We will make the usual assumption that the Poisson algebra of observables is trivial. This is reasonable since the space of states is finite.}. Consider a classical system with a finite set $\Lambda$ of states. To model the phenomenon of \textit{superposition}, consider the complex vector space generated by $\Lambda$ with the usual hermitian product. Now consider the unit sphere $S$ and define the following relation: two vectors in the sphere are equivalent if the differ by a multiplicative global phase, that is $a$ and $b$ in $S$ are equivalent if there is a complex $\lambda$ such that $a=\lambda\,b$. Note that such a complex $\lambda$ must have unit norm, that is it must lie in the unitary group $U(1)$. A quantum state is an equivalence class in the sphere and the corresponding set of quantum sets will be denoted by $q\Lambda$.


Before describing a general quantum network, let us start with a simple case and that is the simple graph $G_1$ with only two nodes. Now, instead of having a two-state Bernoulli variable associated to the edge, we will have a \textit{qubit}. A qubit is the space of equivalence classes in the unit sphere $S^3\subset \C^2$ where two vectors are equivalent if they differ by a global multiplicative phase in $U(1)$. Concretely, $S^3\subset \C^2$ has a natural action by $U(1)$ and the space of orbits is the qubit,
$$\pi:\,S^3\rightarrow q\Lambda_{G_1},\qquad g\cdot (z_1,\,z_2)\,=\,(g\,z_1,\,g\,z_2),\qquad g\in U(1),\quad (z_1,\,z_2)\in \C^2.$$
Now, instead of having just two classical states per bit, the set of quantum states per qubit is the Riemann sphere
$$q\Lambda_{G_1}\,=\,\C P^1\,=\,S^2.$$
A representative in $S^3$ of a quantum state in the qubit has the following general form
\begin{equation}\label{qubit}
\ket{\psi}\,=\,p^{1/2}\,\ket{{\bf 1}}\,+\,q^{1/2}\,z\,\ket{{\bf 0}},\qquad q=1-p,\quad z\in U(1)
\end{equation}
where the expression is given up to some multiplicative global phase. We have used the usual \textit{bra-ket} Dirac notation in quantum mechanics.\footnote{Vectors in a Hilbert space $(H,\,\braket{\ }{\ })$ are denoted as \textit{kets} $\ket{b}$ and functionals in the respective dual space are denoted as \textit{bras} $\bra{a}$. The Riesz dual of $\ket{b}$ is $\bra{b}$. The natural pairing of a bra and a ket is the bracket,
$$\bra{a}\ \ket{b}\,=\,\braket{a}{b}.$$
The joke is on Dirac.
}

Topologically, the relation between quantum states in the qubit and their representatives is described by the non trivial Hopf fibration $\pi:\,S^3\rightarrow\,S^2$. In particular, the non triviality of the fibration implies the non existence of a global continuous choice of representatives on the qubit, that is to say, there is no continuous choice of the multiplicative global phase in $U(1)$ such that the quantum states can be expressed as \eqref{qubit}.

The only connected state of the considered graph we can measure is the classical state ${\bf 1}$ hence the probability of observing this classical state after measurement with the system under the quantum state represented by \eqref{qubit} is \footnote{This is the \textit{Born rule} first formulated by Max Born in \cite{Born} and proved by Andrew M. Gleason in \cite{Gleason} under the assumption of \textit{non contextuality}.}
$$QR_{G_1}\left(\pi(\ket{\psi})\right)\,=\,\vert\braket{{\bf 1}}{\psi}\vert^2\,=\,p.$$
Here $QR$ stands for \textit{quantum reliability}, the function we seek to define. In particular, the quantum state represented by \eqref{qubit} is the quantum analog of the Bernoulli variable with parameter $p$.


Now consider the case of the graph $G_2$ with only two nodes and two edges $a$ and $b$ connecting them. In contrast with the previous case, now we have the phenomenon of \textit{entanglement} that we describe as follows. The space of states of this quantum network is the quotient by the natural action of $U(1)$ on the unit sphere of the hermitian complex vector space generated by the classical states with the usual hermitian product, that is to say
$$\pi:S^7\rightarrow q\Lambda_{G_2}\,=\,S^7/U(1),\qquad S^7\subset \langle\,\ket{{\bf 00}},\, \ket{{\bf 01}},\, \ket{{\bf 10}},\, \ket{{\bf 11}}\,\rangle_\C$$
where the action of the unitary group $U(1)$ on the sphere $S^7$ is by multiplication. Geometrically, the space of quantum states is the complex projective space
$$q\Lambda_{G_2}\,=\, \C P^3.$$
A quantum state represented by $\ket{\psi}$ is \textit{non entangled} if it can be factored out by quantum states of two qubits, one for each edge, that is to say
$$\ket{\psi}\,=\,\ket{\psi_a}\otimes\ket{\psi_b}$$
where $\ket{\psi_a}$ and $\ket{\psi_b}$ have the form \eqref{qubit}. Note that since the tensor product we are working with is balanced by $\C$, this notion is well defined. In particular, the space of non-entangled states is the projective submanifold given by the image of the embedding
$$\C P^1\times \C P^1\hookrightarrow \C P^3,\qquad \left(\pi(\ket{\psi_a}),\,\pi(\ket{\psi_b})\right)\mapsto \pi\left(\ket{\psi_a}\otimes\ket{\psi_b}\right).$$
Since the tensor product we are working with is balanced by $\C$, the embedding is well defined. It is clear that almost every quantum state is entangled since the subspace of non-entangled states has real codimension two. Non entangled states are the quantum analog of independence of classical variables.

The entanglement phenomenon has very counterintuitive effects. As an example, consider the state represented by
$$\ket{\psi_1}\,=\,p^{1/2}\ket{{\bf 11}}\,+\,q^{1/2}\ket{{\bf 00}},\qquad q=1-p.$$
This entangled state has the peculiar phenomenon that measuring the state of one edge automatically determines the state of the other edge in such a way that the second edge copies the state of the measured edge. In particular, the quantum network under this quantum state effectively works as a quantum network with one single edge. Hence, by the previous example with the graph $G_1$, we should have
$$QR_{G_2}\left(\pi(\ket{\psi})\right)\,=\,p.$$

Another very interesting example is provided by the state represented with
$$\ket{\psi_2}\,=\,p^{1/2}\ket{{\bf 01}}\,+\,q^{1/2}\ket{{\bf 10}},\qquad q=1-p.$$
Again, measuring the state of one edge automatically determines the state of the other. However, in contrast with the previous example, the state of the second state is the opposite of the state of the measured edge hence the quantum network under this quantum state behaves as a perfect edge. In particular,
$$QR_{G_2}\left(\pi(\ket{\psi_2})\right)\,=\,1.$$
It is very interesting the fact that the entanglement of two imperfect edges effectively works as a perfect edge.

It is worth to mention that the entanglement phenomenon has other counterintuitive and interesting effects with no classical counterpart like the \textit{no cloning theorem} \cite{Nielsen} and the \textit{monogamy of entanglement} \cite{CKW}, \cite{monogamy}, \cite{OV}.

Now we treat the general case. Consider a graph $G$ and its set of classical states $\Lambda_G$ defined in \eqref{states_classical}. The set of quantum states is the quotient of the unit sphere $S$ in the hermitian space $l^2_\C(\,\Lambda_G\,)$ by the multiplicative action of $U(1)$,
$$q\Lambda_G\,=\,S/U(1),\qquad S\subset l^2_\C(\,\Lambda_G\,).$$
Geometrically, the set of quantum states is the complex projective space
$$q\Lambda_G\,=\,\C P^{2^{\vert E \vert} -1}$$
where $\vert E\vert$ is the number of edges of $G$. Again, it is worth to mention that the relation between the quantum states and the choice of representatives in the sphere $S$ is non trivial in the following sense: topologically, the quotient
$$\pi: S\rightarrow q\Lambda_G$$
is a non trivial fibration with fiber $U(1)$. In particular, there is no global continuous choice of representatives in $S$ for the quantum states in $q\Lambda_G$.

The counterintuitive entanglement effects discussed for the graph $G_2$ extend almost verbatim to a a general graph $G$. Indeed, consider a quantum network modelled on the graph $G$ under the state represented by
$$\ket{\psi}\,=\,p^{1/2}\ket{\zeta}\,+\,q^{1/2}\,z\,\ket{\chi},\qquad q=1-p,\qquad z\in U(1)$$
where $\zeta$ is a connected classical state and $\chi$ is a non connected classical state in $\Lambda_G$. There is an edge $e$ in the graph such that it is active in $\zeta$ and inactive in $\chi$. Then, measuring the state of the edge $e$ automatically determines the states of the other edges hence the quantum network effectively works as the single edge $e$ under the quantum state \eqref{qubit}. In particular, by the previous example with the graph $G_1$,
$$QR_G\left(\,\pi(\ket{\psi})\,\right)\,=\,p.$$

Now we are in position to define the notion of quantum reliability. We seek for a function defined on the set of quantum states taking values in the interval $[0,1]$ such that it extends the classical reliability on non entangled qubits and it reproduces the previous example on entangled states.

\begin{defi}\label{Definition_QR}
Consider a graph $G$ with $n$ edges. A quantum reliability is a smooth function
$$QR_G:\,q\Lambda_G\rightarrow\,[0,1]$$
satisfying the following axioms:
\begin{enumerate}
\item {\bf Extension axiom.} For every $n$-tuple $\ket{\psi_1},\ldots,\,\ket{\psi_n}$ such that
$$\ket{\psi_j}\,=\,p_j^{1/2}\,\ket{{\bf 1}}\,+\,q_j^{1/2}\,z_j\,\ket{{\bf 0}},\qquad q_j=1-p_j,\quad z_j\in U(1),$$
up to some multiplicative global phase, we have
$$QR_G\,\left(\,\pi(\ket{\psi_1}\otimes\ldots\otimes\ket{\psi_n})\,\right)\,=\,R_G\,(p_1,\ldots,\,p_n).$$

\item {\bf Entanglement axiom.} For every connected classical state $\zeta$, every non connected classical state $\chi$ and every $z$ in $U(1)$ we have
$$QR_G\,\left(\,\pi(\ket{\psi})\,\right)\,=\,p,\qquad \ket{\psi}\,=\,p^{1/2}\,\ket{\zeta}\,+\,q^{1/2}\,z\,\ket{\chi},\qquad q=1-p,$$
where $\ket{\psi}$ was written up to some multiplicative global phase.
\end{enumerate}
\end{defi}

The realization of a quantum reliability is the object of the following Theorem which is the first result of the paper and will be proved in section \ref{Section_proof_1}.

\begin{teo}\label{main1}
Consider a graph $G$. There is a unique self adjoint operator
$$\widehat{QR}_G:\,l^2_\C(\,\Lambda_G\,)\rightarrow\,l^2_\C(\,\Lambda_G\,)$$
such that the function $QR_G:\,q\Lambda_G\rightarrow\,[0,1]$ defined by
$$QR_G\left(\,\pi(\ket{\psi})\,\right)\,=\,\bra{\psi}\,\widehat{QR}\,\ket{\psi},\qquad \ket{\psi}\in S,\qquad \pi:S\rightarrow q\Lambda_G$$
is a quantum reliability. Moreover, under the assumption of non contextuality, this quantum reliability value is the probability that the quantum network under the state represented by $\ket{\psi}$ stays connected after measurement.
\end{teo}

We refer to the paper \cite{noncontext} by S. Kochen and E. P. Specker for the notion of non contextuality. It is clear that the function $QR_G$ in Theorem \ref{main1} is well defined since the right hand side defines a $U(1)$ invariant function which factors though $\pi$,
$$\widetilde{QR}_G(\ket{\psi})\,=\,\widetilde{QR}_G(z\,\ket{\psi}),\qquad \widetilde{QR}_G(\ket{\psi})\,=\,\bra{\psi}\,\widehat{QR}\,\ket{\psi},\qquad z\in U(1),\quad \ket{\psi}\in S.$$

The operator in Theorem \ref{main1} will be called the \textit{quantum reliability operator}. We do not know whether there are other quantum reliabilities different from the one in Theorem \ref{main1}. If there were, they would not be induced by a linear operator.

\section{Splitting of quantum reliability}\label{sectionIII}

It is clear from the previous entanglement examples that the simple factorization formula \eqref{SimpleFact} does not work in general for quantum networks. However, there is a splitting formula proved for classical networks that survives in the quantum context. This is the Negami's splitting.

In \cite{Negami} S. Negami introduced the \textit{Negami's polynomial} of a graph and showed that it has a very nice splitting in terms of the partition lattice of a finite set of nodes (Theorem 4.2, \cite{Negami}). Actually, a closer look at his proof shows that he proved a slightly more general version of the splitting specially suited for specializations of the Negami's parameter whereat the splitting might be singular (Theorem 2.1, \cite{Burgos3}).

Variable changes and specializations on the Negami's polynomial produces different other well known graph invariants as the Tutte polynomial, the Reliability polynomial, the Ising model, etc. However, as it was mentioned in the previous paragraph, certain specializations of the Negami's parameter might lead to singular splitting formulas. This is exactly the case for the reliability polynomial. Therefore, a Negami like splitting formula for the reliability polynomial is not immediate.

A Negami like splitting formula for the reliability polynomial were found by F. Robledo and the author in \cite{Burgos1}. One of the main difficulties in this approach was the invertibility of the \textit{connectivity matrix} which is proved by calculating its determinant. This calculation was performed by D. M. Jackson in \cite{Jackson} and independently by the author in \cite{Burgos2}, see also (Section 4, \cite{Burgos4}). For a direct derivation of the splitting as a Negami's specialization, see \cite{Burgos3}.

Now we decribe the splitting. Following Negami's notation, assume that the graph $G$ splits in subgraphs $K$ and $H$ only sharing $m$ common vertices $U=V(K)\cap V(H)$. Let $\Gamma(U)$ denote the set of partitions over $U$ and let $\mathcal{A}=\{U_{1},U_{2},\ldots U_{k}\}$ be one of these partitions. Denote by $K/\mathcal{A}$ and $H/\mathcal{A}$ the graphs obtained by identifying all vertices in each $U_{i}$ of $K$ and $H$ respectively. See Figure \ref{Identification}.

\begin{figure}
\begin{center}
  \includegraphics[width=0.7\textwidth]{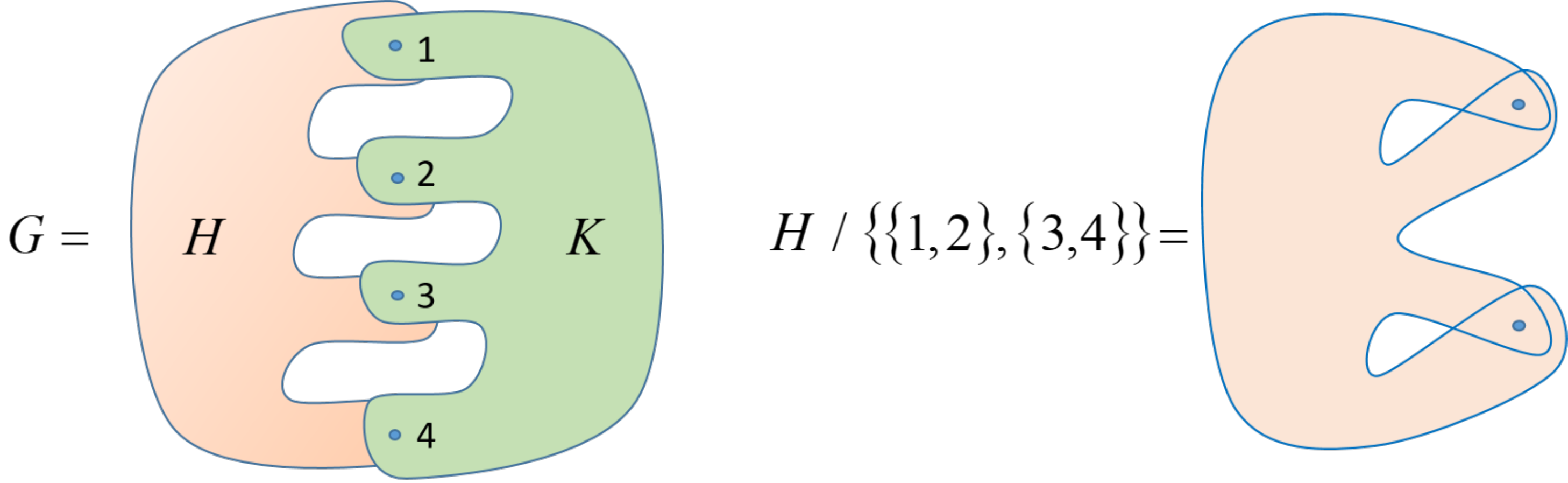}\\
  \end{center}
  \caption{Sharing nodes identification.\ \ \ \ \ \ \ \ \ \ \ \ \ }\label{Identification}
\end{figure}

On  $\Gamma(U)$ define the following partial order: $\gamma\leq \gamma'$ if $\gamma'=\{U'_{1}, \ldots U'_{l}\}$ is a refinement of $\gamma=\{U_{1}, \ldots U_{m}\}$; i.e. For every $U'_{i}$ there is $U_{j}$ such that $U'_{i}\subseteq U_{j}$. The pair $(\Gamma(U), \leq)$ is a lattice. We denote by $\gamma\wedge\gamma'$ the infimum of $\gamma$ and $\gamma'$. Similarly, we denote by $\gamma\vee \gamma'$ the supremum.

Define $\alpha_{\gamma,\gamma'}$ equal to one if $\gamma\wedge \gamma'$ has a single block and equal to zero otherwise. Giving a total ordering on the partitions, these are the entries of the connectivity matrix $M_m$. As it was mentioned before, this matrix is invertible. Define
$$\beta\,=\,(\beta_{\gamma,\gamma'})\,=\,M_m^{-1}.$$

Let $G$ be a graph obtained as a union of two graphs $K$ and $H$ sharing only the vertices $U=\{u_{1},\ldots u_{m}\}$. Then the set of edges decompose as
\begin{equation}\label{Desc_ejes}
E_G\,=\,E_K \sqcup E_H
\end{equation}
hence the set of classical states factors as
\begin{equation}\label{Fact_classical_states}
\Lambda_G\,=\,\Lambda_K \times \Lambda_H
\end{equation}
and we have the factorization
\begin{equation}\label{Fact_L2}
l^2_\C(\Lambda_G)\,=\,l^2_\C(\Lambda_K \times \Lambda_H)\,=\,l^2_\C(\Lambda_K)\,\otimes_\C\, l^2_\C(\Lambda_H).
\end{equation}
Because $E_K\,=\,E_{K/\gamma}$ and $E_H\,=\,E_{H/\gamma'}$ for every pair of partitions $\gamma$ and $\gamma'$ in $\Gamma(U)$, the factorization \eqref{Fact_L2} coincides with the corresponding one induced by $K/\gamma$ and $H/\gamma'$ as well.

The following is the second result of the paper and will be proved in section \ref{Section_proof_2}.

\begin{teo}\label{main2}
Let $G$ be a graph obtained as a union of two graphs $K$ and $H$ sharing only the vertices $U=\{u_{1},\ldots u_{m}\}$. Then,
\begin{equation}\label{equation_main2}
\widehat{QR}_G\,=\,\sum_{\gamma,\,\gamma'\in \Gamma(U)}\ \beta_{\gamma,\gamma'}\ \widehat{QR}_{K/\gamma}\otimes \widehat{QR}_{H/\gamma'},
\end{equation}
where $\widehat{QR}$ is the quantum reliability operator in Theorem \ref{main1}.
\end{teo}

Curiously enough, the proof of Theorem \ref{main2} follows closely the one given in \cite{Burgos1} for classical networks. However, the proofs were rewritten in a shorter, clearer and more elegant way. Except for the proof regarding the connectivity matrix invertibility, the exposition is self contained.

As an example, consider the case of three sharing nodes. With respect to the following ordering of the partition lattice
$$\{\{1\},\{2\},\{3\}\}< \{\{1\},\{2,3\}\}< \{\{1,3\},\{2\}\}< \{\{1,2\},\{3\}\}< \{\{1,2,3\}\}$$
where the minimum and maximum elements are the full and trivial partitions respectively, the connectivity matrix and its inverse read as follows
\begin{equation}\label{Connectivity_three_nodes}
(\alpha_{\gamma,\gamma'})\,=\,\left(
                                                                \begin{array}{ccccc}
                                                                  0 & 0 & 0 & 0 & 1 \\
                                                                  0 & 0 & 1 & 1 & 1 \\
                                                                  0 & 1 & 0 & 1 & 1 \\
                                                                  0 & 1 & 1 & 0 & 1 \\
                                                                  1 & 1 & 1 & 1 & 1 \\
                                                                \end{array}
                                                              \right),\quad
(\beta_{\gamma,\gamma'})\,=\,\frac{1}{2}\left(\begin{array}{ccccc}
1 & -1 & -1 & -1 & 2 \\
-1 & -1 & 1 & 1 & 0 \\
-1 & 1 & -1 & 1 & 0 \\
-1 & 1 & 1 & -1 & 0 \\
2 & 0 & 0 & 0 & 0 \\
\end{array}\right).
\end{equation}
See Figure \ref{Fact3} for an schematic picture of the splitting \eqref{equation_main2} for three sharing nodes.

\begin{figure}
\begin{center}
  \includegraphics[width=0.85\textwidth]{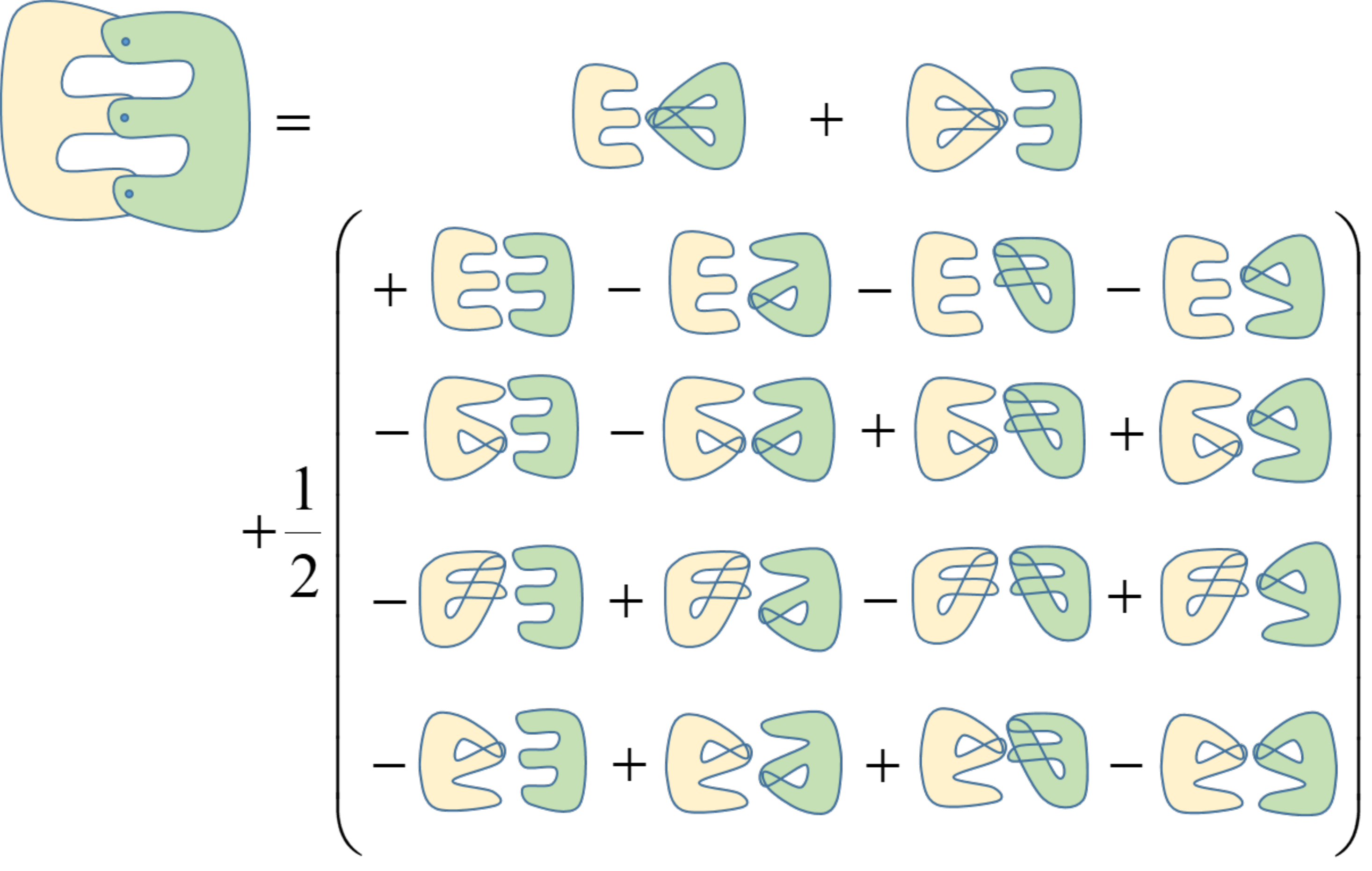}\\
  \end{center}
  \caption{Splitting of the quantum reliability operator for three sharing nodes.}\label{Fact3}
\end{figure}

Because of \eqref{Fact_L2}, we have the following relation among the sets of quantum states: there is a proper embedding
\begin{equation}\label{Fact_quantum_states}
q\Lambda_K\,\times\, q\Lambda_H\,\hookrightarrow\,q\Lambda_G,\qquad \left(\pi(\ket{\psi_K}),\,\pi(\ket{\psi_H})\right)\mapsto \pi\left(\ket{\psi_K}\otimes\ket{\psi_H}\right).
\end{equation}
Again, since the tensor product we are working with is balanced by $\C$, the embedding is well defined. The inclusion is proper as it was expected due to entanglement. It is interesting to compare the relation among quantum states sets \eqref{Fact_quantum_states} with the classical one \eqref{Fact_classical_states}. Again, the left hand side of \eqref{Fact_quantum_states} coincides with the corresponding one induced by $K/\gamma$ and $H/\gamma'$ for every pair of partitions $\gamma$ and $\gamma'$ in $\Gamma(U)$.

The next corollary shows that in contrast with the quantum reliability operator, the induced quantum reliability has a Negami like splitting on the left hand side of \eqref{Fact_quantum_states} and it will not be true in general on the whole quantum states space. The result also shows that the splitting formula for the quantum reliability operator in Theorem \ref{main2} recovers the classical one as it was expected.

\begin{cor}\label{cor1}
Consider the quantum reliability $QR$ induced by the quantum reliability operator as in Theorem \ref{main1}. Under the hypothesis of Theorem \ref{main2}, we have
$$QR_G\left(\,\pi(\ket{\psi_K}\otimes\ket{\psi_H})\,\right)\,=\,\sum_{\gamma,\,\gamma'\in \Gamma(U)}\ \beta_{\gamma,\gamma'}\ QR_{K/\gamma}\left(\,\pi(\ket{\psi_K})\,\right)\, QR_{H/\gamma'}\left(\,\pi(\ket{\psi_H})\,\right),$$
for every $\pi(\ket{\psi_K})$ in $q\Lambda_K$ and every $\pi(\ket{\psi_H})$ in $q\Lambda_H$. In particular, by the extension axiom in definition \ref{Definition_QR}, evaluating on a non entangled state gives the classical Negami's splitting formula. Moreover, in the particular case of two sharing nodes with $K$ being the simple two node graph, the splitting formula recovers the simple factorization formula \eqref{SimpleFact}.
\end{cor}

\section{Hybrid network}\label{sectionIV}

Now, consider a classical network with an imperfect node of degree greater than or equal to three. A natural question is whether we can model this network with the theory of perfect nodes networks described at the beginning of this section, that is to say, whether we can find a network with perfect nodes that effectively works as the original one with the imperfect node. After a moment of thought, the reader may become convinced that this task is impossible within the context of classical networks.

However, in the context of quantum networks this task is achievable and we proceed as follows. Replace the degree $n\geq 3$ imperfect node $v$ with Bernoulli parameter $p$ by the quantum network with perfect nodes consisting of the complete graph $K_n$ as in Figure \ref{quantum_node} under the quantum state represented by
$$\ket{\psi}\,=\,p^{1/2}\,\ket{{\bf 11\ldots 1}}\,+\,q^{1/2}\,z\,\ket{{\bf 00\ldots 0}},\qquad z\in U(1),\qquad q=1-p.$$
The rest of the network remains classic. This is an example of a \textit{hybrid network}. It is clear that, due to the entanglement of the quantum state, this hybrid network effectively works as the original classical network.

\begin{figure}
\begin{center}
  \includegraphics[width=0.85\textwidth]{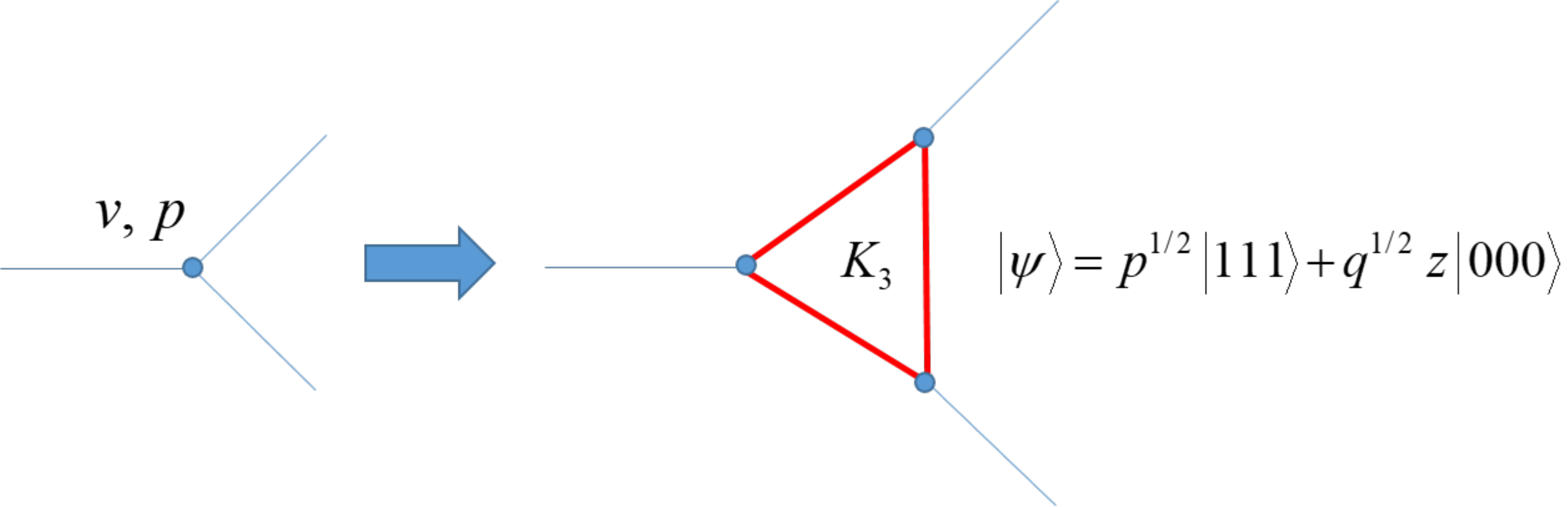}\\
  \end{center}
  \caption{Replacement of an imperfect node by a quantum network with perfect nodes.}\label{quantum_node}
\end{figure}

A hybrid network with perfect nodes canonically has the structure of the graph in Theorem \ref{main2}, that is to say it is canonically the union of two graphs $K$ and $H$ sharing only a finite set of vertices $U$ such that $K$ is a quantum network and $H$ is a classical network. Moreover, these subgraphs are unique with this property. Indeed, $K$ and $H$ are the smallest subgraphs of $G$ whose edges are the quantum and classical edges of $G$ respectively. This structure will be called the \textit{canonical decomposition of $G$ with sharing nodes set $U$}. The set of states of the hybrid network is
\begin{equation}\label{hybrid_states}
h\Lambda_G\,=\,q\Lambda_K\,\times\,\mbox{Conv}\,(\Lambda_H).
\end{equation}
Recall that $\mbox{Conv}\,(\Lambda_H)$ denotes the convex hull of the classical sates $\Lambda_H$ and this is the set of sequences of Bernoulli parameters on the set of edges of the graph $H$.

The quantum reliability of a hybrid state $(\zeta_K,\,(p_e)_{e\in E_H})\in h\Lambda_G$ is defined as the quantum reliability of the quantum state represented by
$$\ket{\Psi}\,=\,\ket{\psi}\otimes\bigotimes_{e\in E_H}\,\left(p_e^{1/2}\,\ket{{\bf 1}_e}\,+\,q_e^{1/2}\,\ket{{\bf 0}_e}\right),\qquad q_e=1-p_e$$
where $\ket{\psi}$ is some unit vector representing the quantum state $\zeta_K$.

From Corollary \ref{cor1} it immediately follows the following splitting formula that allows the calculation of the quantum network reliability of a hybrid network.

\begin{cor}\label{cor3}
Consider the quantum reliability $QR$ induced by the quantum reliability operator as in Theorem \ref{main1}. Consider a hybrid network $G$ under the hybrid sate $\zeta\in h\Lambda_G$ along with its canonical decomposition with sharing nodes set $U$. Denote by $K$ the quantum subgraph and by $H$ the classical subgraph of the decomposition. Then,
$$QR_G\,(\zeta)\,=\,\sum_{\gamma,\,\gamma'\in \Gamma(U)}\ \beta_{\gamma,\gamma'}\,QR_{K/\gamma}\,(\zeta_K)\, R_{H/\gamma'}\,\left(\,(p_e)_{e\in E_H}\right),\quad \zeta=\left(\zeta_K,\,(p_e)_{e\in E_H}\right).$$
\end{cor}

\section{Classical network with a quantum sublayer}\label{sectionVI}

A particular instance of a hybrid network is that of a \textit{classical network with a quantum sublayer} that we now define. See Figure \ref{quantum_sublayer}.

\begin{figure}
\begin{center}
  \includegraphics[width=0.6\textwidth]{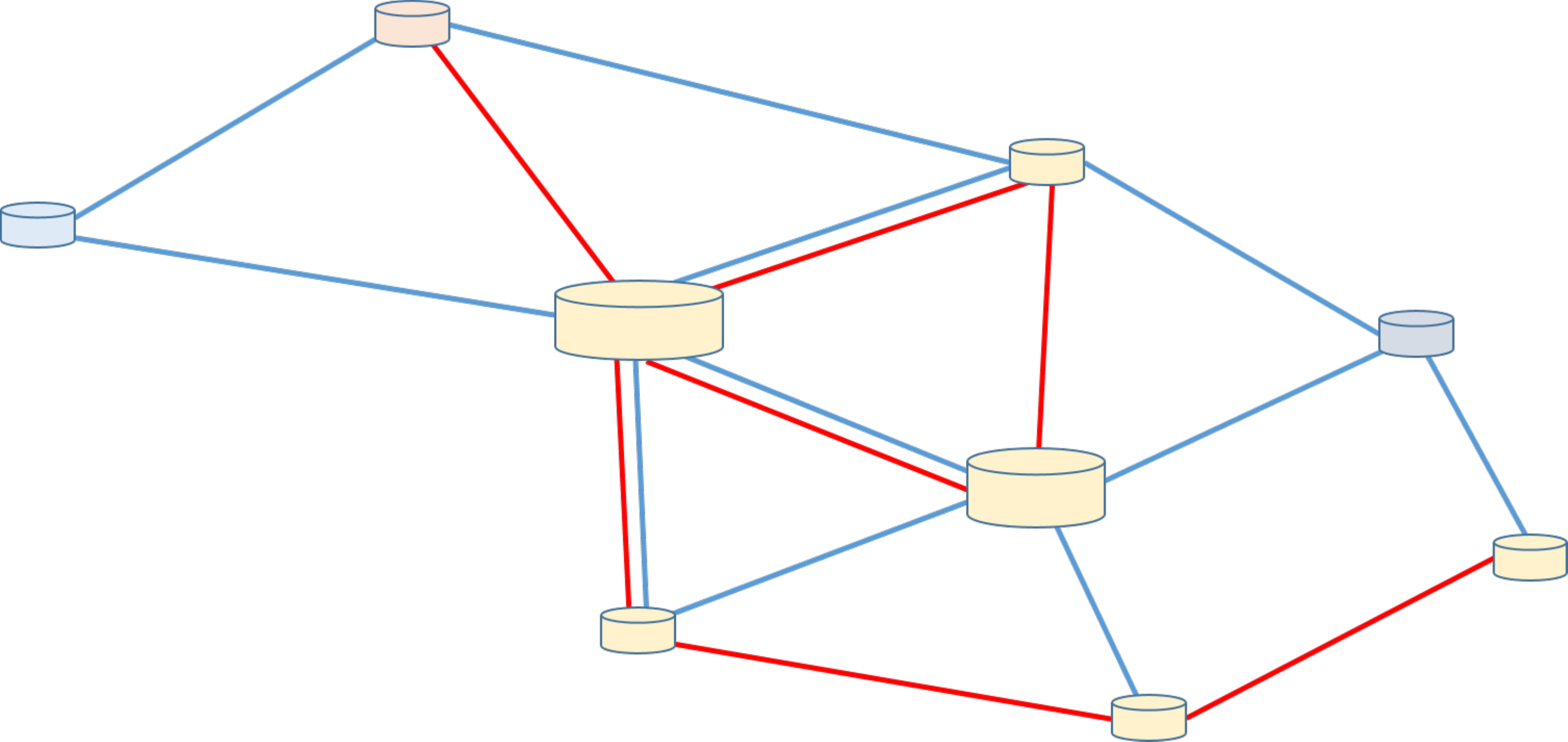}\\
  \end{center}
  \caption{A classical network (in blue) with a quantum sublayer (in red).}\label{quantum_sublayer}
\end{figure}

\begin{defi}
A classical network $H$ with a quantum sublayer $K$ is a hybrid network $G$ whose canonical decomposition consists of the subgraphs $K$ and $H$ such that the set of vertices of $K$ is contained in that of $H$. We will say that $G$ realizes the network $H$ with quantum sublayer $K$.
\end{defi}

In the previous definition, note that the set of sharing nodes of the canonical decomposition of the hybrid network consists of the nodes of the quantum sublayer. The quantum reliability of the hybrid network realizing a classical network with a quantum sublayer is the sum of the reliability of the classical network with quantum corrections caused by the quantum sublayer. Indeed, from Corollary \ref{cor3} we have the following result.

\begin{cor}\label{cor4}
Consider the quantum reliability $QR$ induced by the quantum reliability operator as in Theorem \ref{main1}. Consider a classical network $H$ with a quantum sublayer $K$ realized by the hybrid network $G$. Then,
$$QR_G\,(\zeta)\,=\,R_H\,\left(\,(p_e)_{e\in E_H}\right)\,+\,\sum_{\substack{\gamma,\,\gamma'\in \Gamma(V_K) \\ \gamma\neq {\bf t} }}\ \beta_{\gamma,\gamma'}\,QR_{K/\gamma}\,(\zeta_K)\, R_{H/\gamma'}\,\left(\,(p_e)_{e\in E_H}\right)$$
where $V_K$ is the set of nodes of $K$ and $\zeta\,=\,\left(\zeta_K,\,(p_e)_{e\in E_H}\right)\,\in\, h\Lambda_G$.
\end{cor}

This result is in agreement with the obvious fact that the quantum reliability of a classical network with an inoperative quantum sublayer must coincide with the classical reliability of the classical layer.


\section{Proof of Theorem \ref{main1}}\label{Section_proof_1}

Let $\mathcal{C}\subset \Lambda_G$ be the set of connected states and consider the subspace generated by it $\langle\,\mathcal{C}\,\rangle <l^2_\C(\Lambda_G)$. Define the self adjoint operator $\widehat{QR}_G$ as the orthogonal projection onto $\langle\,\mathcal{C}\,\rangle$,
$$\widehat{QR}_G:\,l^2_\C(\Lambda_G)\rightarrow l^2_\C(\Lambda_G).$$
Concretely, with respect to the basis $\Lambda_G$, $\widehat{QR}_G(\,\ket{\epsilon}\,)\,=\,\ket{\epsilon}$ if $\ket{\epsilon}$ is connected and $\widehat{QR}_G(\,\ket{\epsilon}\,)\,=\,{\bf 0}$ otherwise.

Define the function $\widetilde{QR}_G: S\rightarrow\R$ by the formula
$$\widetilde{QR}_G(\,\ket{\psi}\,)\,=\,\bra{\psi}\,\widehat{QR}_G\,\ket{\psi},\qquad \ket{\psi}\in S$$
where $S\,\subset\, l^2_\C(\Lambda_G)$ is the unit sphere. Because
$$0\,\leq\, \braket{\widehat{QR}_G\,\psi\,}{\,\widehat{QR}_G\,\psi}\,=\,\bra{\psi}\,\widehat{QR}_G\,\ket{\psi}\,\leq\, \braket{\psi}{\psi}=\,1,\qquad \ket{\psi}\in S$$
we have that $\widetilde{QR}_G$ is valued in the interval $[0,1]$. Since
$$\widetilde{QR}_G(\,\mu\ket{\psi}\,)\,=\,\bra{\psi}\bar{\mu}\,\widehat{QR}_G\,\mu\ket{\psi}\,=\,\bar{\mu}\mu\,\bra{\psi}\,\widehat{QR}_G\,\ket{\psi}\,=\,\widetilde{QR}_G(\,\ket{\psi}\,),\qquad \mu\in U(1),$$
the function $\widetilde{QR}_G$ factors through $\pi$ and defines the function $QR_G$, that is to say
$$\xymatrix{ S\ar[rr]^{\widetilde{QR}_G} \ar[d]_\pi & & [0,1] \\
q\Lambda_G \ar[urr]_{QR_G} & & }$$
and the diagram is commutative.

We will show that $QR_G$ is a quantum reliability. Consider the state
$$\ket{\psi}\,=\,\bigotimes_{e\in E}\,\mu_e\,\left(p_e^{1/2}\,\ket{{\bf 1}_e}\,+\,z_e\,q_e^{1/2}\,\ket{{\bf 0}_e}\right),\qquad q_e=1-p_e.$$
After some tedious but elementary calculation, this state can be expressed as
$$\ket{\psi}\,=\,\mu\,\sum_{\epsilon\in \Lambda_G}\,\left(\prod_{e\in E}\,(p_e^{1/2})^{\epsilon(e)}\,(z_e\,q_e^{1/2})^{1-\epsilon(e)}\right)\,\ket{\epsilon},\qquad \mu=\prod_{e\in E}\,\mu_e.$$
In particular, we have
$$\widehat{QR}_G(\,\ket{\psi}\,)\,=\,\mu\,\sum_{\epsilon\in \mathcal{C}}\,\left(\prod_{e\in E}\,(p_e^{1/2})^{\epsilon(e)}\,(z_e\,q_e^{1/2})^{1-\epsilon(e)}\right)\,\ket{\epsilon}$$
and the Riesz dual
$$\bra{\psi}\,=\,\bar{\mu}\,\sum_{\epsilon\in \Lambda_G}\,\left(\prod_{e\in E}\,(p_e^{1/2})^{\epsilon(e)}\,(\bar{z}_e\,q_e^{1/2})^{1-\epsilon(e)}\right)\,\bra{\epsilon}.$$
Because the basis $\{\,\ket{\epsilon}\ |\ \epsilon\in\Lambda_G\,\}$ is orthonormal, we have
$$QR_G\left(\,\pi(\ket{\psi})\,\right)\,=\,\bra{\psi}\,\widehat{QR}_G\,\ket{\psi}\,=\,\sum_{\epsilon\in \mathcal{C}}\,\prod_{e\in E}\,p_e^{\epsilon(e)}\,q_e^{1-\epsilon(e)}\,=\, R_G(p_1,\ldots, p_n).$$
This proves the extension axiom of a quantum reliability.

For the entanglement axiom, consider a connected state $\ket{\zeta}$, a non connected state $\ket{\chi}$ and $z$ in $U(1)$. Then,
$$\left(\bar{\mu}\,\left(p^{1/2}\,\bra{\zeta}\,+\,q^{1/2}\,\bar{z}\,\bra{\chi}\right)\right)\,\widehat{QR}_G\,\left(\mu\,\left(p^{1/2}\,\ket{\zeta}\,+\,q^{1/2}\,z\,\ket{\chi}\right)\right)$$
$$\,=\,\bar{\mu}\mu\,\left(p\,\braket{\zeta}{\zeta}\,+\,(qp)^{1/2}\,\bar{z}\,\braket{\chi}{\zeta}\right)\,=\,p,\qquad q=1-p,$$
since the states are orthonormal. This concludes the proof that $QR_G$ is a quantum reliability.

Suppose that $L:\,l^2_\C(\Lambda_G)\rightarrow l^2_\C(\Lambda_G)$ is a self adjoint operator such that the function $f:\,q\Lambda_G\rightarrow [0,1]$ defined by the commutative diagram
$$\xymatrix{ S\ar[rr]^{\tilde{f}} \ar[d]_\pi & & [0,1] \\
q\Lambda_G \ar[urr]_{f} & & },\qquad \tilde{f}(\,\ket{\psi}\,)\,=\,\bra{\psi}\,L\,\ket{\psi},\qquad \ket{\psi}\in S.$$
is a quantum reliability. By the extension axiom, $\lambda_\epsilon$ equals to one if $\epsilon$ is connected and equals zero otherwise where we have defined
$$\lambda_\epsilon=\bra{\epsilon}\,L\,\ket{\epsilon},\qquad \epsilon\in \Lambda_G.$$

We will show now that the off diagonal elements are zero. Consider a pair of distinct classical states $\epsilon$ and $\epsilon'$ in $\Lambda_G$. Suppose that $c\,=\,\bra{\epsilon}\,L\,\ket{\epsilon'}\,\neq\,0$ and consider the states
$$\ket{\psi_\pm}\,=\,\frac{1}{2^{1/2}}\,\left(\ket{\epsilon}\,\pm\,\frac{\bar{c}}{\Vert c \Vert}\,\ket{\epsilon'}\right).$$
Since $L$ is self addjoint, we have
$$f(\pi(\ket{\psi_\pm}))\,=\,\tilde{f}(\,\ket{\psi_\pm}\,)\,=\,\bra{\psi_\pm}\,L\,\ket{\psi_\pm}\,=\,\frac{1}{2}\,\left(\lambda_\epsilon\,+\,\lambda_{\epsilon'}\right)\,\pm\,\Vert c \Vert.$$
Then, either $f$ takes values greater than one or less than zero or contradicts the entanglement axiom and this is absurd.

We conclude that the operator $L$ coincides with $\widehat{QR}_G$ since
$$L(\ket{\epsilon})\,=\,\sum_{\epsilon'\in\Lambda_G}\,\bra{\epsilon'}\,L\,\ket{\epsilon}\ \ket{\epsilon'}\,=\,\lambda_\epsilon\,=\widehat{QR}_G\,(\ket{\epsilon}),\qquad \epsilon\in\Lambda_G$$
and $\Lambda_G$ is a basis.

Finally, under the assumption of non contextuality, A. M. Gleason \cite{Gleason} proved the Born rule which states that the probability of observing the state $\ket{\epsilon}$ in the system under the state $\ket{\psi}$ after measurement is
$$\mbox{Prob}(\epsilon)\,=\,\vert\,\braket{\epsilon\,}{\,\psi}\,\vert^2.$$
and these events are independent hence
$$\mbox{Prob}\,\left([\epsilon\,\in\, \mathcal{C}]\right)\,=\,\sum_{\epsilon\in \mathcal{C}}\, \mbox{Prob}(\epsilon)\,=\,
\sum_{\epsilon\in \mathcal{C}}\, \vert\,\braket{\epsilon\,}{\,\psi}\,\vert^2\,=\,\,\bra{\psi}\,\widehat{QR}_G\,\ket{\psi} $$
and this concludes the proof.

\section{Proof of Theorem \ref{main2}}\label{Section_proof_2}

Let $G$ be a graph obtained as a union of two graphs $K$ and $H$ sharing only the vertices $U=\{u_{1},\ldots u_{m}\}$.

Given a graph $W$ and a state $\epsilon\in \Lambda_W$, define the graph $W(\epsilon)$ as the graph resulting from removing every edge $e$ from $W$ such that $\epsilon(e)=0$. We will say that a state $\epsilon$ in $\Lambda_W$ is connected if $W(\epsilon)$ is connected and the set of these states will be denoted by $\mathcal{C}(W)$.

Given a state $\epsilon\in \Lambda_H$, a connected component $C$ of $H(\epsilon)$ is an \textit{island in} $H(\epsilon)$ if $C\cap U=\emptyset$. We will say that \textit{a state $\epsilon\in \Lambda_H$ has no islands} if the graph $H(\epsilon)$ has no islands and we will denote the set of these states by $\hat{\Lambda}_H$. Since connected components coincide or are disjoint, the following set is a partition of $U$
\begin{equation}\label{Partition_function}
\mathcal{P}(\epsilon)\,=\,\{\,U\cap C\ |\ C{\rm\ is\ a\ connected\ component\ of\ }H(\epsilon)\,\},\qquad \epsilon\in\hat{\Lambda}_H.
\end{equation}

\begin{lema}\label{Lema1}
For every partition $\gamma$ of $U$, we have the following characterization
$$\mathcal{C}(H/\gamma)\,=\,\{\,\epsilon\in \hat{\Lambda}_H\ |\ \vert\,\mathcal{P}(\epsilon)\wedge\gamma\,\vert=1\,\}.$$
\end{lema}
\begin{proof}
Consider a state $\epsilon\in \hat{\Lambda}_H$ and the map that identifies the vertices of each class of $U$,
$$\varphi:\,H(\epsilon)\rightarrow H(\epsilon)/\gamma.$$
This map is continuous and surjective.

Suppose that $\vert\,\mathcal{P}(\epsilon)\wedge\gamma\,\vert>1$ and consider distinct classes $A$ and $A'$ in $\mathcal{P}(\epsilon)\wedge\gamma$. Consider all the connected components $C_1,\ldots, C_l$ and $C'_1,\ldots, C'_{l'}$ such that
\begin{equation}\label{intersections}
U\cap C_i\,\subset\,A,\qquad U\cap C'_j\,\subset\,A'.
\end{equation}
Since $\epsilon$ has no islands and $\mathcal{P}(\epsilon)$ is a refinement of $\mathcal{P}(\epsilon)\wedge\gamma$, all the intersection in \eqref{intersections} are non empty and fulfill the classes $A$ and $A'$ respectively. By definition of the map $\varphi$, we have that
$$\varphi(C_1)\cup\ldots\cup\varphi(C_l),\qquad \varphi(C'_1)\cup\ldots\cup\varphi(C'_{l'})$$
are distinct connected components of $H(\epsilon)/\gamma$ hence this graph is not connected.

Conversely, suppose that $H(\epsilon)/\gamma$ is not connected and consider distinct connected components $B$ and $B'$ of it. Since $\mathcal{P}(\epsilon)$ is a refinement of $\mathcal{P}(\epsilon)\wedge\gamma$, the preimages $\varphi^{-1}(B)$ and $\varphi^{-1}(B')$ contain connected components of $H(\epsilon)$ and because $\epsilon$ has no islands the sets $U\cap\varphi^{-1}(B)$ and $U\cap\varphi^{-1}(B')$ are non empty distinct classes in $\mathcal{P}(\epsilon)\wedge\gamma$ hence $\vert\,\mathcal{P}(\epsilon)\wedge\gamma\,\vert>1$. This concludes the proof.
\end{proof}

\begin{lema}\label{Lema2}
We have the identity
$$\mathcal{C}(G)\,=\,\bigcup_{\epsilon\in \hat{\Lambda}_H}\,\mathcal{C}\left(\,K/\,\mathcal{P}(\epsilon)\,\right)\times\{\,\epsilon\,\}$$
where the union is disjoint.
\end{lema}
\begin{proof}
Consider a connected state $\epsilon\in \Lambda_G$ and its unique decomposition $(\epsilon_K,\epsilon_H)$ such that $\epsilon_K\in \Lambda_K$ and $\epsilon_H\in \Lambda_H$. Clearly $\epsilon_H$ has no islands. Indeed, if there were an island in $H(\epsilon_H)$, then this would be a proper connected component in $G(\epsilon)$ which is absurd because $\epsilon$ is connected.

Consider the map that identifies each connected component $C$ of the subgraph $H(\epsilon_H)$ into a point corresponding to the class $U\cap C$ in the partition $\mathcal{P}(\epsilon_H)$,
$$\nu:\,G(\epsilon)\rightarrow K(\epsilon_K)/\mathcal{P}(\epsilon_H).$$
This map is continuous and surjective. In particular,
$$1\,\leq \,\vert\,Comp\,(K(\epsilon_K)/\mathcal{P}(\epsilon_H))\,\vert\,\leq\,\vert\,Comp\,(G(\epsilon))\,\vert\,=\,1$$
where $Comp$ denotes the set of connected components of the respective set. Therefore, the graph $K(\epsilon_K)/\mathcal{P}(\epsilon_H)$ is connected.

Conversely, consider a state $\epsilon\in \Lambda_G$ such that $\epsilon_H$ has no islands and $K(\epsilon_K)/\mathcal{P}(\epsilon_H)$ is connected where $(\epsilon_K,\epsilon_H)$ is the unique decomposition described before. Again, the map $\nu$ described before is continuous and surjective. It is clear that $\epsilon_K$ has no islands either since otherwise $K(\epsilon_K)/\mathcal{P}(\epsilon_H)$ would have a proper connected component which is absurd.

We claim that every connected component $A$ of $G(\epsilon)$ satisfies $A\cap U\neq \emptyset$. Indeed, if $A\cap U= \emptyset$, then either $A$ is an island in $H(\epsilon_H)$ or $A$ is an island in $K(\epsilon_K)$ and both alternatives are absurd.

Suppose that $\epsilon$ is not connected and consider distinct connected components $B$ and $B'$ of $G(\epsilon)$. Consider points $a\in B\cap U$ and $b\in B'\cap U$. There is a path $l$ joining $\varphi(a)$ and $\varphi(b)$ in $K(\epsilon_K)/\mathcal{P}(\epsilon_H)$. There is a unique lifting $l'$ in $K(\epsilon_K)$ of the path $l$ by the map $\nu$ joining the points $a'$ and $b'$ in $U$ such that $\varphi(a)=\varphi(a')$ and $\varphi(b)=\varphi(b')$. Hence, there are paths $l_1$ and $l_2$ in $H(\epsilon_H)$ joining $a$ with $a'$ and $b'$ with $b$ respectively. Therefore, the path
$$l_3\,=\,l_1\cdot l'\cdot l_2$$
joins $a$ with $b$ which is absurd since they belong to different connected components. We conclude that $\epsilon$ is connected and this concludes the proof.
\end{proof}

Recall the definition of the \textit{connectivity matrix} $\left(\alpha_{\gamma',\gamma}\right)$ whose entry $\alpha_{\gamma',\gamma}$ equals one if $\gamma\wedge \gamma'$ has a single class and equals zero otherwise.

\begin{prop}\label{Connectivity_matrix}
The connectivity matrix is invertible.
\end{prop}
\begin{proof}
This is the content of the work by the author in \cite{Burgos2}. It also follows as a corollary of the calculation by D. M. Jackson in \cite{Jackson}. See also (Section 4, \cite{Burgos4}).
\end{proof}

\begin{proof}[Proof of Theorem \ref{main2}]

Given a graph $W$ and a state $\epsilon\in \Lambda_W$, define the orthogonal projection on the subspace generated by $\ket{\epsilon}$,
$$P_{\,\,\ket{\epsilon}}:\,l^2_\C(\,\Lambda_W\,)\rightarrow l^2_\C(\,\Lambda_W\,).$$
Concretely, with respect to the basis $\Lambda_W$, $P_{\,\,\ket{\epsilon}}(\ket{\epsilon})=\ket{\epsilon}$ and $P_{\,\,\ket{\epsilon}}(\ket{\epsilon'})=0$ for any other state $\epsilon'\in \Lambda_W$ distinct from $\epsilon$.

From the proof of Theorem \ref{main1} in section \ref{Section_proof_1}, we can express the quantum reliability operator in term of the projections as
\begin{equation}\label{QR_proj}
\widehat{QR}_{\,W}\,=\,\sum_{\epsilon\in \mathcal{C}(W)}\,P_{\,\,\ket{\epsilon}}.
\end{equation}

Now consider a state $\epsilon\in \Lambda_G$ and its unique decomposition $(\epsilon_K,\epsilon_H)$ such that $\epsilon_K\in \Lambda_K$ and $\epsilon_H\in \Lambda_H$. Then,
\begin{equation}\label{Proj_desc}
P_{\,\,\ket{\epsilon}}\,=\,P_{\,\,\ket{\epsilon_K}\,\otimes\, \ket{\epsilon_H}}\,=\,P_{\,\,\ket{\epsilon_K}}\,\otimes\,P_{\,\,\ket{\epsilon_H}}.
\end{equation}

Given a partition $\gamma$ of $U$, define the orthogonal projection onto the subspace of $\Lambda_H$ generated by the set of those states having no islands whose associated partition is $\gamma$,
$$\hat{\I}_{\gamma}:\,l^2_\C(\,\Lambda_H\,)\rightarrow l^2_\C(\,\Lambda_H\,),\qquad \hat{\I}_{\gamma}\,=\,\sum_{\substack{\epsilon\in \hat{\Lambda}_H \\ \mathcal{P}(\epsilon)=\gamma}}\,P_{\,\,\ket{\epsilon}}.$$

From Lemmas \ref{Lema1}, \ref{Lema2} and identities \eqref{QR_proj}, \eqref{Proj_desc} we immediately have the identities
\begin{equation}\label{Identities}
\widehat{QR}_G\,=\,\sum_{\gamma\in \Gamma(U)}\,\widehat{QR}_{K/\gamma}\,\otimes\,\hat{\I}_\gamma,\qquad \widehat{QR}_{H/\gamma'}\,=\,\sum_{\gamma\in \Gamma(U)}\,\,\alpha_{\gamma',\gamma}\,\hat{\I}_\gamma.
\end{equation}

By Proposition \ref{Connectivity_matrix}, the second identity in \eqref{Identities} can be inverted and substituted in the first giving the desired splitting. This concludes the proof.
\end{proof}

\section{Proof of Corollaries \ref{cor1}, \ref{cor3} and \ref{cor4}}

\begin{proof}[Proof of Corollary \ref{cor1}]
This immediately follows from the following calculation
$$QR_G\left(\,\pi(\ket{\psi_K}\otimes\ket{\psi_H})\,\right)\,=\,\bra{\psi_K}\otimes\bra{\psi_H}\, \widehat{QR}_G\,\ket{\psi_K}\otimes\ket{\psi_H}$$
\begin{eqnarray*}
&=&\,\sum_{\gamma,\,\gamma'\in \Gamma(U)}\ \beta_{\gamma,\gamma'}\ \bra{\psi_K}\otimes\bra{\psi_H}\,\widehat{QR}_{K/\gamma}\otimes \widehat{QR}_{H/\gamma'}\,\ket{\psi_K}\otimes\ket{\psi_H} \\
&=&\,\sum_{\gamma,\,\gamma'\in \Gamma(U)}\ \beta_{\gamma,\gamma'}\ \bra{\psi_K}\,\widehat{QR}_{K/\gamma}\,\ket{\psi_K}\ \bra{\psi_H}\,\widehat{QR}_{H/\gamma'}\,\ket{\psi_H} \\
&=&\,\sum_{\gamma,\,\gamma'\in \Gamma(U)}\ \beta_{\gamma,\gamma'}\ QR_{K/\gamma}\left(\,\pi(\ket{\psi_K})\,\right)\, QR_{H/\gamma'}\left(\,\pi(\ket{\psi_H})\,\right).
\end{eqnarray*}
\end{proof}

\begin{proof}[Proof of Corollary \ref{cor3}]
This immediately follows from the definition of a hybrid network, Corollary \ref{cor1} and the extension axiom in definition \ref{Definition_QR}.
\end{proof}

\begin{proof}[Proof of Corollary \ref{cor4}]
Since $\alpha_{{\bf t},\gamma}=1$ and $\alpha_{{\bf f},\gamma}=\delta_{{\bf t}, \gamma}$ for every partition $\gamma$, it is clear that
$$\beta_{{\bf t},\gamma}\,=\,\delta_{{\bf f},\gamma},\qquad\sum_{\gamma'\in\Gamma(U)}\,\beta_{\gamma,\gamma'}\,=\,\delta_{{\bf t},\gamma},\qquad \gamma\in\Gamma(U),$$
where ${\bf t}$ and ${\bf f}$ denote the trivial and full partitions respectively and $\delta$ denotes the Kronecker delta. See the matrices \eqref{Connectivity_three_nodes} for a concrete example.

Now, the result immediately follows from Corollary \ref{cor3} and the fact that the function $QR_{K/{V_K}}$ is identically one since every state in $\Lambda_K$ is connected in the graph $K/{V_K}$.
\end{proof}

\end{document}